\documentclass[12pt,reqno,a4paper]{amsart}
\usepackage{amsmath,amssymb,amsfonts,amscd}
\usepackage[mathscr]{eucal}
\usepackage{comment}
\usepackage{hyperref}

\date{19 September (2 October) 2013}

\author{H.~M.~Khudaverdian}
\address{School of Mathematics,  University of Manchester, Oxford Road,  Manchester,   M13 9PL,  UK}
\email{khudian@manchester.ac.uk, theodore.voronov@manchester.ac.uk}
\author{Th.~Th. Voronov}
\title[Geometric constructions on the algebra of densities]{Geometric constructions on the algebra of densities}

\newtheorem{theorem}{Theorem}
\newtheorem{proposition}{Proposition}
\newtheorem{corollary}{Corollary}
\theoremstyle{definition}
\newtheorem{definition}{Definition}%[section]
\newtheorem{example}{Example}%[section]
\newtheorem{remark}{Remark}%[section]

\def\co{\colon\thinspace}

\renewcommand{\leq}{\leqslant}

\newcommand{\wed}{\wedge}

\renewcommand{\div}{\mathop{\mathrm{div}}}
\newcommand{\divrho}{\mathrm{\div}_{\!\boldsymbol{\rho}}}

\DeclareMathOperator{\Ber}{Ber} \DeclareMathOperator{\Ker}{Ker}

 \DeclareMathOperator{\sign}{sign}
 \DeclareMathOperator{\grad}{grad}
 \DeclareMathOperator{\Res}{Res}

\DeclareMathOperator{\Vol}{Vol} 
 
\DeclareMathOperator{\Der}{Der}

\DeclareMathOperator{\Diff}{Diff}

\newcommand{\fun}{C^{\infty}}

\newcommand{\der}[2]{{\frac{\partial {#1}}{\partial {#2}}}}

\newcommand{\lder}[2]{{\partial {#1}/\partial {#2}}}

\newcommand{\RR}{\mathbb R}

\newcommand{\p}{\partial}

\newcommand{\w}{{\mathbf{w}}}

\renewcommand{\a}{\alpha}

\newcommand{\e}{\varepsilon}

\renewcommand{\O}{\Omega}
\newcommand{\D}{\Delta}

\newcommand{\g}{{\gamma}}
\newcommand{\G}{{\Gamma}}

\newcommand{\lt}{\theta} % last term in the master Hamiltonian for long bracket

\newcommand{\la}{{\lambda}}

\renewcommand{\d}{\delta}

\newcommand{\at}{{\tilde a}}
\newcommand{\bt}{{\tilde b}}

\newcommand{\Xt}{{\tilde X}}

\newcommand{\Yt}{{\tilde Y}}

\newcommand{\Pt}{{\tilde P}}

\newcommand{\ps}{{\boldsymbol{\psi}}}
\newcommand{\ph}{{\boldsymbol{\varphi}}}

\newcommand{\ch}{{\boldsymbol{\chi}}}

\newcommand{\X}{{\boldsymbol{X}}}
\newcommand{\Y}{{\boldsymbol{Y}}}

\newcommand{\bXt}{{\tilde \X}}
\newcommand{\bYt}{{\tilde \Y}}

\newcommand{\XX}{\mathfrak{X}}
\newcommand{\YY}{\mathfrak{Y}}
\newcommand{\SSS}{\mathfrak{S}}
\newcommand{\PP}{\mathfrak{P}}
\newcommand{\QQ}{\mathfrak{Q}}

\newcommand{\YYt}{{\tilde \YY}}
\newcommand{\XXt}{{\tilde \XX}}

\newcommand{\rh}{{\boldsymbol{\rho}}}

\DeclareMathOperator{\Vect}{\mathrm{Vect}}
\DeclareMathOperator{\SVect}{\mathrm{SVect}}
\DeclareMathOperator{\Ver}{\mathrm{Vert}}

\newcommand{\A}{\mathfrak{A}}

\DeclareMathOperator{\SA}{\mathrm{S}\mathfrak{A}}

\DeclareMathOperator{\Den}{\mathfrak{F}}
\begin{document}

\begin{abstract}
The  algebra of densities  $\Den(M)$ is a commutative algebra canonically associated with a given  manifold or supermanifold $M$. We introduced this algebra earlier in connection with our studies of Batalin--Vilkovisky geometry. The algebra $\Den(M)$ is graded by real numbers and possesses a natural invariant scalar product. This  leads to important geometric consequences and  applications   to geometric constructions on the original manifold. In particular, there is a classification theorem for derivations of the algebra $\Den(M)$.
It allows {a natural definition of}
bracket operations on vector densities of various weights on a (super)manifold $M$,
similar to how the classical Fr\"{o}licher--Nijenhuis theorem on derivations of the algebra of differential forms leads to the Nijenhuis bracket.  It is possible to  extend  this classification from ``vector fields''   (derivations) on $\Den(M)$   to ``multivector fields''. This leads to the striking result that an arbitrary even Poisson structure on $M$ possesses a canonical lifting to the algebra of densities.  (The latter two statements were obtained by our student A.~Biggs.) This is in sharp contrast with the previously studied case of an odd Poisson structure, where extra data are required for such a lifting.
\end{abstract}

\maketitle

\tableofcontents

\section{Introduction}

The paper is devoted to natural differential-geometric constructions on the algebra of densities, which is a commutative algebra canonically associated with a given  manifold or supermanifold. It is   based on results of  the authors  and includes a recent result due to our   student A.~Biggs.

A \emph{density} of \emph{weight} $\lambda$ is a geometric object which in local coordinates has the form $f(x)|Dx|^{\lambda}$. Here $\lambda$ is an arbitrary real number. We assume that the coefficients $f(x)$ are smooth. It is clear that densities can be multiplied so that their weights are added. The resulting commutative \emph{algebra of densities}, which we denote here by $\Den(M)$, is graded by real numbers and possesses a natural invariant scalar product arising from integration of densities of weight $+1$ (i.e., volume forms).

Densities of weight $0$ are just functions on manifolds, so $\fun(M)$ is contained in $\Den(M)$. Hence $\Den(M)$ has a unit $1$ (a constant function).

The commutative algebra $\mathfrak F(M)$ can be identified with a subalgebra of the algebra  of smooth functions $\fun(\hat M)$ on an ``extended'' manifold $\hat M$,  which is  the total space of a fiber bundle $\hat M\to M$ with   one-dimensional fibers\,\footnote{\,The extended manifold $\hat M$ is sometimes referred to as the ``\,Thomas bundle\,'' of a manifold $M$. In
{the} 1920s, T.~Y.~Thomas studied the construction of $\hat M$ in relation with
{the} projective theory of linear connections on manifolds. To him belongs a striking result~\cite{thomas:projtheory} that projective classes of symmetric linear connections on $M$ are in one-to-one correspondence with linear connections on the manifold $\hat M$.   The  link  with Thomas's work was found by our student J.~George~\cite{jgeorge1, jgeorge2}.}. More precisely, $\hat M$ is the frame bundle for the line bundle $\det TM$. (In the super case we of course must use the notation $\Ber TM$.)  The grading and the  scalar product are peculiar {to}
the subalgebra $\mathfrak F(M)$ and  do not extend to the whole algebra $\fun(\hat M)$. We have inclusions of algebras
\begin{equation*}
    \fun(M)\subset \Den(M) \subset \fun(\hat M)\,.
\end{equation*}

{The invariant   scalar product on the algebra $\Den(M)$ can be alternatively perceived  as a generalized volume element, in the sense of generalized functions~\cite{gelfand:shilov}, on the manifold $\hat M$. As   a linear functional it is defined on the subspace $\Den(M)\subset \fun(\hat M)$ rather than on the whole $\fun(\hat M)$. Nevertheless it can be expressed   by a conventionally-looking analytic formula. } Its existence    leads to important geometric consequences. %\comm{Suggest mention explicitly that `generalized' is in sense of Gelfand-Graev, etc. There is a similar comment later in paper; don't need both.}

Namely, for the Lie algebra of graded vector fields on $\hat M$\,---\,or the derivations of $\Den(M)$\,---\,there arises a canonical divergence operator. In particular, one can speak about the divergence-free derivations. Similarly,  for the  algebra of graded multivector fields on $\hat M$, there exists a canonical odd Laplacian or a Batalin-Vilkovisky type operator, so it is possible to consider those fields that are annihilated by this operator.
From this we can obtain classification theorems and deduce constructions of natural brackets on $M$, as follows.

\begin{itemize}
  \item There is  a one-to-one correspondence between the divergence-free derivations of $\mathfrak F(M)$ of   weight $\lambda\neq 1$ and   the vector densities of the same weight   on $M$. (For weight zero, this is the relation between the Lie derivatives and the corresponding vector fields on $M$.)
  \item A similar statement holds for the  graded multivector fields on $\hat M$ annihilated by the odd Laplacian and  the multivector densities on $M$.
\end{itemize}

One can deduce that  the commutator of vector fields  and the  Schouten bracket  of multivector fields  naturally extend to brackets of vector densities and multivector densities, respectively. Another application of this correspondence is the possibility of lifting  an arbitrary even Poisson structure on $M$ to the algebra of densities $\mathfrak F(M)$, without any additional structure.

These constructions for the algebra of densities have a similarity with the classical Fr\"{o}licher--Nijenhuis   theorem on derivations of the algebra of forms and the construction of the Nijenhuis bracket that follows from it. We can see the Fr\"{o}licher--Nijenhuis   theorem and the constructions for the algebra $\Den(M)$ as two particular instances of    ``second-order geometry'', i.e.,     geometry arising from an  iteration of   natural first-order constructions   {such as, e.g., taking the tangent bundle\,\footnote{{The idea of second-order geometry has been stressed by K.~C.~H.~Mackenzie~\cite{mackenzie:secondorder1, mackenzie:secondorder2} in connection with his studies of double vector bundles, double Lie groupoids and double Lie algebroids.}}.}
Indeed, in the Fr\"{o}licher--Nijenhuis setup the starting point is the   antitangent bundle $\Pi TM$,  for which    functions     are   forms on $M$, while for us here the starting point is the Thomas bundle $\hat M$, the frame bundle of $\Ber TM$, for which  functions   are densities on $M$. Both $\Pi TM$ and $\hat M$ are of the first order   with respect to $M$. We then   study first-order objects  (e.g., derivations) on them, which will be of  the second order relative to $M$. Therefore, classification theorems such as quoted above or that of Fr\"{o}licher--Nijenhuis give  us information about objects living on $M$ through   raids into second-order geometry.

A few words about the origins of the algebra $\Den(M)$ and the related constructions discussed in this paper.

In 1989 (published in 1991), one of us (H.~M.~Kh.)~\cite{hov:deltabest} gave an invariant geometric construction for the odd Laplace-type operator introduced by I.~A.~Batalin and G.~A.~Vilkovisky~\cite{bv:perv, bv:vtor} as the key tool in their quantization method for gauge systems. This construction is of the form ``$\div \grad f$'' where ``$\grad f$'' stands for the hamiltonian vector field of a function $f$ with respect to a given odd symplectic (or odd Poisson) structure and $\div=\div_{\rho}$ is a divergence of vector fields defined by a choice of a volume element, on which therefore the construction depends (the crucial fact here is that, in contrast with the usual case, for an \emph{odd} symplectic structure there is no invariant ``Liouville'' measure). Another odd Laplacian was discovered by H.~M.~Khudaverdian in~\cite{hov:max} (see also~\cite{hov:proclms, hov:semi}). It acts on semidensities on odd symplectic manifolds and   is canonical in the sense that it does not require any additional structure. The odd Laplacian on semidensities is probably even more fundamental for the Batalin--Vilkovisky geometry than the odd Laplacian on functions.

An analysis given by both present authors in~\cite{tv:laplace1} clarified the distinguished role of semidensities for odd Poisson geometry in general (in particular, we discovered there a groupoid property of the Batalin--Vilkovisky equation). This naturally led us  to studying Laplace-like operators on arbitrary densities~\cite{tv:laplace2}. As we discovered in~\cite{tv:laplace2},   {the introduction}
of the algebra $\Den(M)$ and   differential operators (in the algebraic sense) acting on it,  instead of   operators acting on  spaces of densities of isolated weights, allows {a complete classification of} the Batalin--Vilkovisky or Laplace-type operators in this setting. Namely, there is a one-to-one correspondence between the self-adjoint BV-type operators on $\Den(M)$ normalized by the condition $\Delta 1=0$ and the corresponding `brackets' (which arise as their principal symbols). Here the statement  does not depend on whether the operators are even or odd, and the `brackets' in question are symmetric\,\footnote{In the case of odd brackets and odd Laplacians, there is a   further investigation in~\cite{tv:laplace2} of  conditions equivalent to Jacobi identities. Extending the results of~\cite{tv:laplace2} in another direction, a certain `groupoid of connections' generalizing  the Batalin--Vilkovisky groupoid of~\cite{tv:laplace1} was introduced and studied in~\cite{tv:hovgroupoids}.}. Therefore the results equally apply to supermanifolds or ordinary manifolds.

The classification of derivations of the algebra $\Den(M)$ was one of the   results in~\cite{tv:laplace2}. In~\cite{tv:laplace2}, it was somewhat buried under many other important results, in particular, those related   with the second-order operators.  Recently, our student A.~Biggs has generalized this result to the corresponding version of multivector fields~\cite{biggs:lifting}. His point of departure was an  observation  obtained by bare hands  that, unlike \emph{odd} Poisson brackets (the case studied by us in~\cite{tv:laplace2}), \emph{even} Poisson brackets on a manifold or supermanifold naturally   extend to densities without requiring any additional geometric data. The statement about multivector fields on $\hat M$   provides an explanation for this initial beautiful observation.

The purpose of this short survey is to introduce the reader to the algebra of densities $\Den(M)$ and geometric constructions related to it, particularly, the derivations and multivector fields\,\footnote{There is a deep relation between geometry and algebra related with $\Den(M)$ and the studies of intertwining operators for $\Diff M$-modules and equivariant quantization in the works of Duval--Ovsienko~\cite{duval:ovsienko1997} and Lecomte~\cite{lecomte:proj1999}.   See also~\cite{hov:pencils}}.

\section{Preliminaries. Densities and the algebra formed by them}
\subsection{Volume elements}
Let $M$ be a manifold or supermanifold. We shall employ the following notation: $x^a$ denote local coordinates (in the case of a supermanifold, even and odd together) and $Dx$ stands for the coordinate volume element, which under a change of coordinates $x^a=x^a(x')$ transforms as $Dx=(Dx/Dx')\,Dx'$, where $Dx/Dx'$ is the Berezinian of the Jacobi matrix $\lder{x^a}{x^{a'}}$. We shall refer to this Berezinian\footnote{For an ordinary manifold, this is of course the usual determinant.} as to the Jacobian of the coordinate transformation. Let $|Dx|$ be a symbol transforming according to the rule $|Dx|=|Dx/Dx'|\,|Dx'|$, as the notation suggests. (We introduce absolute values in order to avoid difficulties with non-integer powers.)

For practical  calculation   of the effect of changes  of coordinates, it is useful to express  the coordinate volume element $Dx$ as  $[dx]$, i.e., in greater detail, as $[dx^1,\ldots,dx^n]$, for an ordinary manifold (here all the coordinates are even variables) and $[dx^1,\ldots,dx^n\,|\, d\xi^1,\ldots,d\xi^m]$, for a supermanifold (where  the first $n$ coordinates are even and the last $m$ coordinates are odd). Here it is understood that the differentials $dx^a$ form a local frame for  the cotangent bundle\,\footnote{More precisely, the cotangent bundle     $T^*M$ or the anticotangent bundle $\Pi T^*M$, depending on conventions on parity, but this difference is not important here.} and the meaning of the `square bracket' operation is as follows. For a vector space or a free module over a commutative (super)algebra, it is the function of a basis which in the ordinary case is just the full ($n$-fold, $n$ being the dimension) exterior product of the basis elements, $[e_1,\ldots,e_n]=e_1\wed\ldots\wed e_n$; in the super case, it is a symbol which is multiplied by the Berezinian  if the basis undergoes an invertible linear transformation. For calculations, it is sufficient to use the following properties  (which  define the bracket symbol uniquely):
\begin{itemize}
  \item homogeneity: if a basis element is multiplied by an invertible factor, then the bracket is multiplied by the same factor in the power $+1$ for an even basis element and in the power $-1$ for an odd basis element;
  \item invariance under elementary transformations: when a basis element is replaced by the sum with another   element with a coefficient of the appropriate parity.
\end{itemize}
One can quickly learn that it is as convenient to make calculations with the symbol $[dx^1,\ldots,dx^n\,|\, d\xi^1,\ldots,d\xi^m]$ on supermanifolds as with the exterior product $dx^1\wed \ldots \wed dx^n$ on ordinary manifolds.

\subsection{Recollection: densities}
\begin{definition} A (smooth) \emph{density} of \emph{weight} $\la$ is a geometric object which in local coordinates has the form $\psi(x)|Dx|^{\lambda}$. (We assume that the coefficient $\psi(x)$ is smooth.) Equivalently, it is a smooth section of the line bundle $|\Vol(M)|^{\otimes \la}$, where $\la\in\RR$.
\end{definition}

The line bundle $|\Vol(M)|^{\otimes \la}$, by definition, has   local frames $|Dx|^{\la}$ associated with coordinate systems on $M$  with the   transformation law $|Dx|^{\la}=|Dx/Dx'|^{\la}\,|Dx'|^{\la}$.

Notation: $\Den_{\la}(M)= \{\text{all densities of weight $\la$ on $M$}\}$.

\begin{remark}[on the parity of $|Dx|$] There are different conventions as to which parity should be assigned to the coordinate volume element $Dx$ (the natural options are $n$ or $n+m$ modulo $2$ if $\dim M =n|m$). Respectively, depending on the dimension of the supermanifold  and adopted convention, the  `line' bundle $\Vol (M)=\Ber T^*M$ has rank $1|0$ or $0|1$. (The same holds for the dual bundle $\Ber TM=\Vol (M)^*$.) However, there is little choice for the symbol $|Dx|$ if we wish to consider its various powers. We have to agree that $|Dx|$ is always \emph{even}. So the bundles $|\Vol (M)|$ and $\Vol (M)^*$ have rank $1|0$.
\end{remark}

\subsection{Multiplication of densities. The algebra $\Den(M)$}
Densities are multiplied in the obvious way: for $\ps=\psi(x)|Dx|^{\la}$ and $\ph=\varphi(x)|Dx|^{\mu},$
\begin{equation}
    \ps\ph=\psi(x)\varphi(x) |Dx|^{\la+\mu}\,.
\end{equation}
(This corresponds to the natural isomorphism $|\Vol(M)|^{\otimes \la}\otimes |\Vol(M)|^{\otimes \mu}=|\Vol(M)|^{\otimes (\la+\mu)}$\,.)
\begin{definition}
Consider formal finite sums $\sum_{\la} \psi_{\la}(x)|Dx|^{\la}$ and extend the multiplication to them by allowing to open brackets. We arrive at an associative algebra
\begin{equation}\label{algebra}
    \Den(M)=\bigoplus_{\la}\Den_{\la}(M)\,,
\end{equation}
which we call  the \emph{algebra of densities} on a (super)manifold $M$.
\end{definition}
Obvious properties: $\Den(M)$ is a commutative $\RR$-graded algebra;  the algebra of smooth functions $\fun(M)$ is contained in $\Den(M)$ as a subalgebra (functions are of course densities of weight zero); in particular, the algebra $\Den(M)$ has a unit $1$ (a constant function).

\subsection{An invariant scalar product on the algebra $\Den(M)$}
\begin{definition}
For compactly-supported densities $\ps\in\Den_{\la}(M)$, $\ch\in \Den_{\mu}(M)$, define their \emph{scalar product} as
\begin{equation}\label{scalarproductfirst}
    (\ps,\ch):=\begin{cases}
    \int_{M} \ps\ch=\int_M\psi(x)\chi(x)|Dx|\quad &\text{if $\la+\mu=1$}\,,\\
    0 \quad &\text{otherwise}\,.
    \end{cases}
\end{equation}
It is extended by linearity to arbitrary compactly-supported elements of $\Den(M)$.
\end{definition}
Properties: non-degeneracy (which follows from the non-degeneracy of the Berezin integral) and invariance, i.e.,
\begin{equation}
    (\ps\ph,\ch)=(\ps,\ph\ch)\,.
\end{equation}
Therefore
\begin{equation}
    (\ps,\ch)=(\ps\ch,1)\,.
\end{equation}
 The existence of such an invariant scalar product on the algebra of densities leads to important consequences.

{\small

\begin{remark}The possibility to integrate densities of weight $+1$, i.e., objects that in local coordinates have the form $\psi(x)|Dx|$, in the case of supermanifolds requires an orientation condition. Namely, one needs to fix an orientation of the normal bundle $N=N_{M/M_0}$ to the carrier of $M$, in other words, an `orientation in the odd directions'. Recall that the formula for the change of variables in the Berezin integral involves the sign of the determinant of the even-even block of the Jacobi matrix only, therefore on supermanifolds objects of the form $f(x)D_{1,0}x$, not of the form $f(x)|Dx|$, can be integrated without extra orientation conditions. (Here the symbol $D_{1,0}x$ transforms with the factor $\Ber J \cdot \sign \det J_{00}$, where $J$ is the Jacobi matrix.) See more about different orientations in the super case in~\cite{tv:pdf, tv:coh} and \cite{tv:git}. We shall not speak about this subtlety any further.
\end{remark}

}

In view of the invariance property, it is a matter of taste whether to speak about the scalar product $(\ps,\ch)$ of two densities or about the linear form $I(\ps):=(\ps,1)$ on $\Den(M)$, which we shall refer to as to the \emph{formal integral} on $\Den(M)$ and which is just the ordinary integral extended to formal sums of densities of various weights (by setting its value to zero on densities of weight $\neq +1$).

\subsection{Interpretation of densities as functions}
Consider an element $\ps$ of the algebra $\Den(M)$, so that in local coordinates $x^a$,
\begin{equation}
    \ps= \sum\psi_{\la}(x)|Dx|^{\la}\,.
\end{equation}
It is convenient to replace $|Dx|$ by a formal variable $t$, which is assumed to be invertible. In this way we assign a `generating function' to a density $\ps$\,:
\begin{equation}
    \ps(x,t):=\sum\psi_{\la}(x)t^{\la}\,.
\end{equation}
The formal variable $t$ has the following transformation law under a change of coordinates: $t=t' \left|\frac{Dx}{Dx'}\right|$.

Note that the functions of the variable $t$ that we consider here are of a very special type. We  call a function $f(t)$  \emph{pseudo-polynomial} if it is a finite linear combination of powers $t^{\la}$, where the exponents $\la$ can be arbitrary real numbers. (The reader should compare with various classes of symbols arising in the theory of pseudodifferential operators.) The generating functions of the elements of $\Den(M)$ are pseudo-polynomials w.r.t. $t$.

The description of densities by generating functions has  a direct geometric meaning. Namely, the variables $x^a,t$ (where $t\neq 0$) can be considered as local coordinates on an ``extended'' manifold $\hat M$, where $\dim \hat M=\dim M+1$. There is a natural fiber bundle structure $\hat M\to M$. The bundle $\hat M$ is nothing but the frame bundle for the line bundle $|\Ber TM|$ (or $|\det TM|$ in the case of ordinary manifolds). Indeed, the variable $t$ stands for $|Dx|$, which is a basis section of $|\Ber T^*M|$ and, respectively, a linear function on $|\Ber TM|$.  As explained in the Introduction, this manifold plays a role in projective geometry of    linear connections on $M$ and is sometimes referred to as the ``Thomas bundle'' of $M$. We will not explore this relation here.

From now on, we identify densities $\ps$ with their generating functions $\ps(x,t)$. Therefore, densities can be viewed as functions on the manifold $\hat M$. We arrive at an embedding of $\Den(M)$ into $\fun(\hat M)$. It  needs to be stressed that the algebra $\Den(M)$ regarded as a subalgebra of $\fun(\hat M)$ has special properties distinguishing its elements from arbitrary functions on $\hat M$. The key difference is   the existence of grading on $\Den(M)$, which is not defined on the whole of $\fun(\hat M)$. In the following we shall speak about various geometric objects on the manifold $\hat M$, but   always confining ourselves to   graded objects. The reader will see how the possibility to use grading makes a difference.

That means, although we won't make it precise, that  we  actually treat $\hat M$ as a formal `graded manifold' and  consider  the $\RR$-graded algebra  $\Den (M)$ (and \emph{not} the full algebra $\fun(\hat M)$) as its `algebra of functions'. This is a certain departure from the classical viewpoint. In particular, we shall not assign any numerical values to the variable $t$, treating it completely formally and assuming only that $t^{-1}$ makes sense.

Note also, to avoid confusion, that the $\RR$-grading of the algebra $\Den (M)$ has nothing to do with parity and has no influence on the commutativity rules.

\subsection{A convenient expression for the scalar product}
As said, we identify the elements $\ps\in \Den(M)$ with the corresponding functions $\ps(x,t)$. Using that, we can re-write the definition of the scalar product as follows: for (compactly-supported) $\ps, \ch\in \Den(M)$,
\begin{align}
    (\ps,\ch)&=\int_{M}\!|Dx|\;\Res_0 \bigl(t^{-2}\ps(x,t)\ch(x,t)\bigr) \label{scalarproductres}\\
    & = \int_{\hat M} \!|D(x,t)|t^{-2}\; \ps(x,t)\ch(x,t)\,. \label{scalarproductformint}
\end{align}
Here step~\eqref{scalarproductres} is obvious: taking the residue at zero in $t$ (after adjusting the powers by dividing by $t^2$) serves to single out the term of weight $+1$ from our formal sum.

As for the next step~\eqref{scalarproductformint}, here the residue  at zero followed by the integration over $M$ is interpreted as a (formal) integral over the graded manifold $\hat M$.
This requires more  {explanation, which we provide} below.

\subsection{The invariant scalar product on $\Den(M)$ as a generalized volume element on the graded manifold $\hat M$}

Let us consider the expression $|D(x,t)|t^{-2}$ from the viewpoint of its transformation law. By writing $D(x,t)t^{-2}=[dx,dt] t^{-2}$ we obtain
\begin{multline*}
    [dx,dt]=[\,dx^{a'}\der{x}{x^{a'}}\,, dx^{a'}\der{t}{x^{a'}}+ dt' \der{t}{t'}\,]=
    [\,dx', dx^{a'}\der{t}{x^{a'}}+ dt' \der{t}{t'}\,]\, \Ber \der{x}{x'}=\\
    [\,dx', dt' \der{t}{t'}\,]\, \Ber \der{x}{x'}=
    [dx', dt']\;\der{t}{t'}\,\, \Ber \der{x}{x'}   =   [dx', dt']\,\Bigl(\Ber \der{x}{x'}\Bigr)^{\!2}\,,
\end{multline*}
since
\begin{equation*}
    t=t'\, \Ber \der{x}{x'}\,,
\end{equation*}
and therefore
\begin{equation}
    [dx,dt]t^{-2}= [dx', dt'] t'^{-2}\,.
\end{equation}
Hence $D(x,t)t^{-2}$ (and even more so $|D(x,t)|t^{-2}$) is invariant under changes of coordinates. We conclude that the manifold $\hat M$ is endowed with a canonical volume element, which is expressed as $|D(x,t)|t^{-2}$ in arbitrary local coordinates $x^a,t$ on $\hat M$.

Now, let us compare the volume element on $\hat M$ just obtained with the formal integral $I(\psi)$ or the invariant scalar product $(\psi,\chi)$ on the algebra
$\Den (M)\subset \fun(\hat M)$.
We claim that
these structures are equivalent. This equivalence manifests itself in two ways.

On one hand, as claimed in the previous subsection, there is a way of identifying the scalar product  or the formal integral of densities $I(\psi)$ with an integral against the volume element $|D(x,t)|t^{-2}$. This we will explain in a moment.

Indeed, consider functions of $t$ where $t$ is invertible. In which sense can they   be integrated against  {$(Dt)\,t^{-2}$}. A simpler question is about integration against $Dt$. As a guiding principle we take the requirement  that $\int\! Dt \,\der{\ps}{t}=0$. On the class of pseudo-polynomials we immediately conclude that  $\int\! Dt \, \ps(t)$ must be proportional to the residue of $\ps(t)$   at $0$ (since all  powers of $t$ except for $t^{-1}$ are derivatives of  functions in this class). Thus it is legitimate to set $\int\! Dt :=\Res_0$. Combining this with the ordinary integration over $M$,  we arrive   at the expression for the scalar product in the form~\eqref{scalarproductformint} {for the elements of $\Den (M)$.}

On the other hand, we defined in~\cite{tv:laplace2}  a canonical  divergence of graded vector fields on $\hat M$ (i.e, derivations of $\Den(M)$)  with the help of the invariant scalar product on $\Den(M)$;  as  we shall later see, this divergence is given by the standard formula (for  a divergence relative a volume element) if one uses the volume element $|D(x,t)|t^{-2}$.

We speak of $|D(x,t)|t^{-2}$ as of a `generalized' volume element because the integration against it is defined in a formal fashion as a linear functional on a particular space of functions. (In the case interesting for us, this is the space of pseudo-polynomials in $t$.)

\begin{remark} If we introduce a new variable $u$  so that $u=t^{-1}$, the volume element $|D(x,t)|t^{-2}$ takes the simple form $|D(x,u)|$. Instead of the expansion  over $t$, we may use the expansion over $u$, as
\begin{equation}
    \ps(x,u)=\sum \frac{\psi_{\la}(x)}{u^{\la}} \quad \text{for} \quad \ps= \sum \psi_{\la}(x)|Dx|^{\la}\,.
\end{equation}
The geometric meaning of $u$, which corresponds to $|Dx|^{-1}$,  is of course the local basis element for $|\Ber TM|=\Vol (M)^*$.
\end{remark}

\subsection{Constructions with $\Den(M)$\,: algebraic and geometric viewpoints}
Individual spaces of densities   $\Den_{\la}(M)$ for particular weights are standard objects of study in mathematics. The novelty introduced by our approach is not in taking all $\Den_{\la}(M)$   together, but in exploiting the commutative algebra structure on their direct sum~\eqref{algebra}
\begin{comment}\begin{equation*}
    \Den (M)=\bigoplus_{\la} \Den_{\la}(M)\,,
\end{equation*}
\end{comment}
and, in particular, using  the invariant scalar product on $\Den(M)$.  Considering the algebra $\Den(M)$ instead of individual spaces  of densities  {prompts the application of} standard algebraic notions such as differential operators and derivations.
Besides that, we may use the geometric interpretation of the elements of the commutative algebra $\Den(M)$ as functions on the graded manifold $\hat M$. This allows  {us} to view algebraic constructions on  $\Den(M)$ as differential-geometric constructions on $\hat M$.

The ultimate goal of course is to  re-interpret the outcome  in terms of objects on the initial (super)manifold $M$.

\subsection{Example: differential operators on $\Den(M)$}
Recall that differential operators can be introduced algebraically. For a commutative algebra, they are defined by induction. The operators of order zero are those commuting with the multiplication operators. The operators of order $\leq n+1$ are those commuting with the multiplication operators modulo operators of order $\leq n$. (This can be generalized to
non-commutative and
non-associative settings, but we do not need  {to do} that.) For differentiable manifolds or supermanifolds, this reproduces the usual notion.

Differential operators on the algebra $\Den(M)$, defined algebraically, may be treated  as (graded) differential operators
in the usual sense on the manifold $\hat M$. Working in local coordinates, we may expand them in partial derivatives $\lder{}{x^a}$ and $\lder{}{t}$. Note that the derivatives $\lder{}{x^a}$ do not change  weights of objects, but the derivative $\lder{}{t}$ decreases weights by one. Hence   it is more convenient instead of $\lder{}{t}$ to consider the operator $\w=t\lder{}{t}$, which has weight zero. (Since $t$ is invertible, passing from $\lder{}{t}$ to $t\lder{}{t}$ is harmless.) The operator $\w$ is nothing but the \emph{weight operator}, with the eigenvalue $\la$ on the eigenspace $\Den_{\la}(M)$.

Therefore, a differential operator  of degree  $N$ and weight $\mu$ on the algebra $\Den(M)$ (or, equivalently, on the manifold $\hat M$) has the local expression
\begin{equation}\label{operator}
    L=\sum_{|\a|+k\leq N} a_{\a k}(x,t)\partial_{x}^{\,\a}\w^k
\end{equation}
where $\a$ is a multi-index and the shorthand $\partial_{x}^{\,\a}$ has the usual meaning. All the coefficients $a_{\a k}(x,t)$ have weight $\mu$.

Note that in~\eqref{operator}, $\w$ commutes with $\partial_{x}^{\,\a}$, but not with the coefficients $a_{\a k}(x,t)$.

It will be helpful to recall the behavior of partial derivatives under a change of local coordinates. If on $\hat M$, $x^a=x^a(x')$ and   $t= t'\, J^{-1}$, where $J=|Dx'/Dx|$, then
\begin{align*}
    \der{}{x^a}&= \der{x^{a'}}{x^a}\, \der{}{x^{a'}} + t'\,J^{-1}\p_aJ\, \der{}{t'}\,,\\
    \der{}{t} &= J \, \der{}{t'}\,.
\end{align*}
By re-writing  {these}
in terms of $t\lder{}{t}$ and $t'\lder{}{t'}$, we obtain finally
\begin{align}
    \der{}{x^a}&= \der{x^{a'}}{x^a}\, \der{}{x^{a'}} +  (\p_a\ln J)\, t'\der{}{t'}\,, \label{translawforda}\\
    t\der{}{t} &= t'\der{}{t'} \equiv \w\,.\label{translawfordt}
\end{align}
We see once again the invariance of the operator $\w$ and we also observe that the partial derivative $\p_a=\lder{}{x^a}$ has an additional term in the transformation law, which is proportional to  $\w$.

\subsection{Operators on $\Den(M)$ and operator pencils. Adjoint operator}
When a differential operator $L$  of weight $\mu$ on the algebra $\Den(M)$,
 \begin{equation}
    L\co \Den(M)\to \Den(M)\,,
 \end{equation}
is restricted to the subspaces $\Den_{\la}(M)$,  with varying $\la$,  it translates into a one-parameter family of differential operators on $M$ or an `operator pencil'
\begin{equation}
   L_{\la}\co \Den_{\la}(M)\to \Den_{\la+\mu}(M)\,.
\end{equation}

The operator $\w$ becomes $\la$ upon restriction on $\Den_{\la}(M)$.
Note also the transformation law
\begin{equation}\label{translawonlambda}
    \p_a = \p_ax^{a'}\, \p_{a'} + \la \, \p_a\ln J
\end{equation}
on $\Den_{\la}$, as follows from~\eqref{translawforda}.

\begin{example}
Operators of   order $\leq 1$   on $\Den(M)$ correspond to linear  pencils:
\begin{equation}
    L_{\la}=|Dx|^{\mu}\Bigl(A^a(x)\,\p_a +  B(x) + \la\, C(x)\Bigr)= L^{(0)}+\la\, L^{(1)}\,.
\end{equation}
Here $L^{(0)}=|Dx|^{\mu}\Bigl(A^a(x)\,\p_a +  B(x)\Bigr)$, obtained by setting $\la=0$, makes sense as a usual differential operator of order $\leq 1$   {sending}
functions to $\mu$-densities. (The part $|Dx|^{\mu}\,A^a(x)\,\p_a$ is a vector density of weight $\mu$ and the part $|Dx|^{\mu}B(x)$ is a scalar density.) In contrast with that, the coefficient $L^{(1)}$ does not have an independent meaning and its transformation law involves $A^a$. (One can see that in new coordinates $C'=J^{-\mu}(C+A^a\p_a\ln J)$, by using~\eqref{translawonlambda}.)
\end{example}

\begin{example}
Operators of   order   $\leq 2$ on $\Den(M)$ correspond to   quadratic pencils:
\begin{equation}
    L_{\la}=L^{(0)} + \la L^{(1)} + \la^2 L^{(2)}\,,
\end{equation}
where $L^{(0)}$ is an operator of order $\leq 2$ in the ordinary sense, $L^{(1)}$ contains only the first derivatives in $x^a$, and $L^{(2)}$ does not contain differentiation in $x^a$. Like the above, $L^{(0)}$ can be seen as a differential operator of order $\leq 2$ on $M$  acting from functions to $\mu$-densities, while the operator coefficients $L^{(1)}$ and $L^{(2)}$ are coordinate-dependent and their transformation law involves $L^{(0)}$.
\end{example}

Due to the existence of the canonical scalar product on $\Den(M)$, it makes sense to speak about the \textbf{adjoint operator}  $L^*$ for an operator $L\co \Den(M)\to \Den(M)$. For example, $\w^*=1-\w$. (We shall use this later on.)

% \comml{End of K's comments for September 29, 2013}

\section{Derivations of the algebra $\Den(M)$ and their classification}

\subsection{General form of a vector field on $\hat M$}
From now on without further indication we shall consider only graded geometric objects on the graded manifold $\hat M$: either homogeneous in the sense of weight or finite sums of homogeneous objects. We shall use the letter $w$ for the weight of a homogeneous object, viz., $w(\ps)=\la$ if $\ps\in \Den_{\la}(M)$.

Recall that we have introduced the vector field $\w=t\lder{}{t}$, which has the same form in all coordinates systems $x^a,t$. We have $w(\w)=0$ and $\w \ps = w(\ps) \ps$ for a homogeneous density $\ps$.

A arbitrary vector field $\X$ on $\hat M$ (= a derivation of $\Den(M)$) has the following general form in local coordinates:
\begin{equation}\label{vectfield}
    \X=X^a(x,t)\,\der{}{x^a}+ X^0(x,t)\,\w\,.
\end{equation}
We have $w(\X)=w(X^a)=w(X^0)$.

We use boldface to distinguish vector fields on $\hat M$ from vector fields on the original manifold $M$ (which we denote with the usual font).

\begin{example} Consider a vector field $X$ on $M$. It generates transformations of all geometric objects on $M$ including densities. Since the multiplication of densities is preserved by diffeomorphisms of $M$, the corresponding infinitesimal generator, which is the Lie derivative $L_X$, is a derivation of the algebra $\Den(M)$. Equivalently, the Lie derivative $L_X$ is a vector field on $\hat M$. If $X=X^a(x)\der{}{x^a}$, we have
\begin{equation}\label{lieder}
    L_X=X^a(x)\,\der{}{x^a} + \p_aX^a(-1)^{\at(\Xt+1)}\,\w\,.
\end{equation}
Note that $w(L_X)=0$\,.
\end{example}

We shall shortly introduce a large class of vector fields on $\hat M$ generalizing the Lie derivatives $L_X$.

\subsection{Canonical divergence}% of  vector fields on $\hat M$}
Due to the existence of the invariant scalar product on $\Den(M)$ (or the invariant volume element $|D(x,t)|t^{-2}$ on $\hat M$), there is a canonical {divergence} operator $\div$ on $\Vect(\hat M)$\,,
\begin{equation}\label{divergence}
    \div \co \Vect(\hat M)\to \Den(M)\,.
\end{equation}
There are two equivalent ways of defining the operation $\div$. One is in an algebraic fashion.  %making use of the invariant scalar product on the algebra $\Den(M)$.
(This is how we introduced it in~\cite{tv:laplace2}.) The other is by a standard differential-geometric formula.%, where  the canonical volume element on $\hat M$ is used.
We shall describe both approaches and show that they give the same.

Let us recall some general facts.

From an algebraic viewpoint, an abstract  `divergence' or `divergence operator' for a commutative (super)algebra $A$ (see Koszul~\cite{koszul:crochet85} and Kosmann-Schwarzbach \& Monterde~\cite{yvette:divergence}) is an even linear mapping $\d\co \Der A\to A$ satisfying
\begin{equation*}
    \d(aX) = a\,\d(X) + (-1)^{\at\Xt} X(a)\,,
\end{equation*}
where $a\in A$, $X\in \Der A$. Such divergence operators are in one-to-one correspondence with covariant derivatives (Koszul connections) on the module of volume forms for $A$. The   formula linking the two notions is
\begin{equation}
    \d(X)= - \nabla_X + L_X\,,
\end{equation}
where at the l.h.s., $\d(X)\in A$ is a divergence of $X\in \Der A$ and at the r.h.s., $\nabla_X$ and $L_X$ are a covariant derivative and the Lie derivative along $X$, respectively.
\medskip

{\small
Indeed, ignoring   difficulties related with possible infinite-dimensionality, we may treat a volume form $\xi$ for $A$ as an element of the dual space, which we denote $A'$. An element  $X\in \Der A$ defines the Lie derivative $L_X$ on $A'$  by the formula $\langle L_X \xi, a\rangle= - \langle  \xi, X(a)\rangle (-1)^{\tilde\xi\at}$. From here one can see that on volume forms the Lie derivative satisfies
\begin{equation*}
    L_X(a\xi)=X(a)\xi+(-1)^{\Xt\at}aL_X\xi\,,
\end{equation*}
as it should, and also an additional identity
\begin{equation*}
    L_{aX}(\xi)=(-1)^{\at\Xt} L_X(a\xi)=
   aL_X\xi+ (-1)^{\at\Xt} X(a)\xi\,.
\end{equation*}
Suppose $\nabla$ is a connection on the $A$-module $A'$ defined by the usual axioms for $\nabla_X$. Then the differential operator $\d(X):=-\nabla_X + L_X$ on $A'$ satisfies
\begin{align*}
    \d(X)(a\xi)&=(-1)^{\Xt\at} a\,\d(X)(\xi)\,, \quad \text{and} \\
    \d(aX)(\xi)&=a\,\d(X)(\xi) + (-1)^{\at\Xt} X(a)\xi\,.
\end{align*}
Hence it is an operator of order zero and if we could show that it is the operator of multiplication by an element of $A$, which we may also denote $\d(X)$, this will define a divergence $\Der A\to A$. An operator $T$ of order zero on $A'$ defines the dual operator $T'$ acting on $A''$, which is also of order zero; if $A=A''$ and assuming that there is a unit $1\in A$, we may conclude that the dual operator and hence the original operator on $A'$ are operators of multiplication by an element $a_T\in A$ defined from the identity $\langle \xi, a_T \rangle= \langle T(\xi), 1\rangle (-1)^{\tilde T\tilde\xi}$.  In our particular case   the identity takes the form  $\langle \xi, \d(X) \rangle= \langle -\nabla_X(\xi), 1\rangle (-1)^{\tilde X\tilde\xi}$ (because $L_X'=-X$ kills $1$). The difficulty is that this identity may not define an element of $A$. However, if it does, this is the required mapping from connections on $A'$ to divergence operators. The mapping in the opposite direction is obtained by setting $\nabla_X:= L_X - \d(X)$ on $A'$ for a given divergence operator $\d$. Then the axioms of a covariant derivative are immediately satisfied.

}

\smallskip
In a differential-geometric setup, when $A=\fun(M)$  for a (super)manifold $M$, the one-to-one correspondence between divergences on $\Vect(M)$ and connections in the line bundle $\Vol(M)$ is expressed by the   formulas:
\begin{equation}\label{divgamma}
    \d(X)= (\p_a +\g_a)X^a (-1)^{\at(\Xt+1)} \,,
\end{equation}
for a divergence of vector fields, and
\begin{equation}%\label{ }
    \nabla_X \bigl(\rho\,Dx\bigr)=  X^a(\p_a-\g_a)\rho  \,Dx\,,
\end{equation}
for a covariant derivative of volume forms,
where  the coefficients $\g_a$ are the same\,\footnote{\,Here  $-\g_a$ are the connection coefficients for a connection in   $\Vol (M)$. In~\cite{tv:laplace2}, we used the opposite sign for $\g_a$. If $\G_{ab}^c$ are the Christoffel symbols of a linear connection on $M$, one can choose $\g_a=\G_{ab}^b(-1)^{\bt(\at+1)}$.}.

Suppose now that  the algebra $A$ is endowed with   an invariant scalar product. This permits an identification of the vector spaces $A'$ and $A$ as   $A$-modules. The Lie derivative   $L_X$ becomes under this identification the negative of the adjoint operator  w.r.t. the given scalar product, $L_X=-X^*$. Now, as a  covariant derivative $\nabla_X$ one can take the standard action of $X\in \Der A$. In this way we arrive at the formula
\begin{equation} \label{divfromscprod}
    \d(X)= -(X+X^*)
\end{equation}
used in~\cite{tv:laplace2}. One can check   that  in the differential-geometric setup,  if a scalar product on $\fun(M)$ for a supermanifold $M$ is defined as
\begin{equation*}
    (f,g)=\int_M\! |Dx|\, \rho(x)\, f(x)g(x)\,,
\end{equation*}
for some  volume element $\rh=\rho(x)|Dx|$, then the divergence defined by~\eqref{divfromscprod} is given by the familiar expression $\d(X)= \divrho X$\,,
\begin{equation} \label{divrho}
    \divrho X=\frac{1}{\rho}\,\p_a(\rho X^a)(-1)^{\at(\Xt+1)}\,.
\end{equation}
(This is of course a special case of~\eqref{divgamma}, with $\g_a=\p_a\ln \rho$.)

Therefore the algebraic approach to divergence  is equivalent to the usual  differential-geometric concept when both are applicable.

We shall apply these ideas to our graded manifold $\hat M$ and the algebra $\Den(M)$ regarded as the algebra of functions on $\hat M$. The algebra $\Den(M)$ possesses   the canonical invariant scalar product given by~\eqref{scalarproductfirst},\eqref{scalarproductres},\eqref{scalarproductformint}. %The algebraic approach leads to the following definition.

\begin{definition}[\cite{tv:laplace2}]
The (canonical) \emph{divergence} of a graded vector field $\X\in \Vect \hat M$ is
\begin{equation}\label{divcan}
    \div \X := -(\X+\X^*)\,,
\end{equation}
where $\X^*$ is the adjoint of $\X$ regarded as a differential operator on the algebra $\Den(M)$ w.r.t. the canonical scalar product on $\Den(M)$.
\end{definition}
(The r.h.s. of~\eqref{divcan} is  an operator of zeroth order, hence an element of $\Den(M)$\,.)

Let us calculate an explicit formula for the divergence. Suppose $\X\in \Vect (\hat M)$ in local coordinates is given by~\eqref{vectfield}. Calculating the adjoint operator, we arrive at
\begin{equation*}
    \X^*\ps= -\p_a( X^a \ps)(-1)^{\at(\bXt+\at)}+(1-\w)(X^0\ps)
\end{equation*}
(recall that $\p_a^*=-\p_a$ and $\w^*=1-\w$ as already noted). Expanding this further  gives
\begin{equation*}
    \X^* = -\p_aX^a(-1)^{\at(\bXt+\at)}- X^a\p_a   +(1-\w)(X^0) - X^0\w\,.
\end{equation*}
Therefore we have arrived at the following statement.
\begin{proposition}
The canonical divergence on vector fields on $\hat M$ (or derivations of $\Den(M)$) in local coordinates is given by the formula
\begin{equation}\label{divexplic}
\boxed{\
    \div \X =  \p_aX^a(-1)^{\at(\bXt+1)} + (\w-1)(X^0)\,.
    \vphantom{\der{}{y}}}
\end{equation}
\end{proposition}

\begin{corollary} Suppose $\X\in \Vect(\hat M)$ has weight $\mu$ and in local coordinates
\begin{equation}\label{vectfieldmu}
    \X=t^{\mu}\bigl(X^a(x)\,\p_a+ X^0(x)\,\w\bigr)\,.
\end{equation}
%(so $X^a$ and $X^0$ have a slightly different meaning from the above).
Then
\begin{equation}\label{divexplicmu}
    \div \X =  t^{\mu}\left(\p_aX^a(-1)^{\at(\bXt+1)} + (\mu-1)X^0\right)\,.
\end{equation}
\end{corollary}

\begin{remark}
Since the scalar product on $\Den(M)$ corresponds to an  invariant volume element on $\hat M$, we should expect, from the general theory outlined above, that the `algebraic' definition of  divergence on $\Vect(\hat M)$   as $-(\X+\X^*)$ gives the same answer as  the   `differential-geometric' definition based on  $|D(x,t)|t^{-2}$. That this is indeed true can be verified directly: denoting temporarily the divergence of $\X$ w.r.t. the volume element $|D(x,t)|t^{-2}$ by $\div{'}$, we obtain
\begin{multline*}
    \div{'} \X= \frac{1}{t^{-2}}\left(\der{(t^{-2}X^a(x,t))}{x^a}(-1)^{\at(\bXt+1)}+\der{(t^{-2}tX^0(x,t))}{t}\right)=\\
    \p_a X^a(-1)^{\at(\bXt+1)} + t^2 \der{}{t}(t^{-1}X^0)=
    \p_a X^a(-1)^{\at(\bXt+1)} + t^2\left(-t^{-2}X^0 + t^{-1}\der{X^0}{t}\right)=\\
    \p_a X^a(-1)^{\at(\bXt+1)} +  \left(- X^0 + t \der{X^0}{t}\right) =
    \p_a X^a(-1)^{\at(\bXt+1)} +  \left(\w-1\right)(X^0)\,,
\end{multline*}
which coincides with~\eqref{divexplic}.
\end{remark}

\begin{example}
Consider a vector field on $\hat M$ of the form $\X=L_X$, where $X\in\Vect(M)$. By~\eqref{lieder}, we have
\begin{equation*}
    \X=X^a(x)\,\p_a + \p_aX^a(-1)^{(\Xt+1)\at}\,\w\,,
\end{equation*}
so $X^0=\p_aX^a(-1)^{(\Xt+1)\at}$\,. Substituting that into~\eqref{vectfieldmu}, with $\mu=0$, we arrive at
\begin{equation*}
    \div L_X =   \p_aX^a(-1)^{\at(\bXt+1)} -X^0= \p_aX^a(-1)^{\at(\bXt+1)} -\p_aX^a(-1)^{(\Xt+1)\at}=0 \,.
\end{equation*}
We see that for the Lie derivative $L_X$ along an arbitrary vector field $X\in \Vect(X)$ on $M$, the canonical divergence $\div L_X$ automatically vanishes. This is an expression of the invariance of the scalar product on $\Den(M)$ and the volume element $|D(x,t)|t^{-2}$ on $\hat M$ under the diffeomorphisms of $M$\,.
\end{example}

\begin{example}\label{divvert}
Consider a `vertical' vector field $\X=X^0(x,t)\w$. (Note that here the coefficient $X^0$ is invariant under changes of coordinates.)
Then $\div \X= (\w-1)(X^0)$\,. In particular, if $w(\X)=\mu$, so that $\X=t^{\mu}X^0(x)\w$, we obtain $\div \X=(\mu-1) t^{\mu}X^0(x)$ and thus $t^{\mu}X^0(x)=\frac{1}{\mu-1}\div\X$ for $\mu\neq 1$.
\end{example}

\subsection{Classification of derivations of $\Den(M)$}
Using the notion of canonical divergence we are able to classify all derivations of $\Den(M)$ or vector fields on $\hat M$, except for weight $\mu=1$.

First we consider  the derivations of $\Den(M)$ with vanishing divergence. They   allow the following complete description.
\begin{theorem} \label{classdivfree}
Every divergence-free vector field $\X\in\Vect(\hat M)$   of weight $\mu\neq 1$ is uniquely defined by its restriction on the subalgebra $\fun(M)\subset \Den(M)$, which can be  an arbitrary vector density $\XX\in \Vect(M,\Den_{\mu})$    of weight $\mu$. In local coordinates, $\XX=|Dx|^{\mu}\,\XX^a(x)\p_a$ and %the component $X^0$ of $\X$ is defined by $X^0=(1-\mu)^{-1}\left(\p_aX^a(-1)^{(\Xt+1)\at}\right)$ and
\begin{equation}\label{divfree}
    \X=t^{\mu}\left(\XX^a(x)\,\p_a -  (\mu-1)^{-1}\,\p_a\XX^a(-1)^{\at(\XXt+1)}\,\w\right)\,.
\end{equation}
\end{theorem}
\begin{proof} Indeed, consider a vector field $\X$ of weight $\mu$ on $\hat M$. In local coordinates,
\begin{equation*}
    \X=t^{\mu}\bigl(X^a(x)\,\p_a+ X^0(x)\,\w\bigr)\,.
\end{equation*}
Assume that $\div \X=0$. By formula~\eqref{divexplicmu}, this is equivalent to the equation
\begin{equation*}
     \p_aX^a(-1)^{\at(\bXt+1)} + (\mu-1)X^0=0\,.
\end{equation*}
If $\mu\neq 1$, this can be uniquely solved for $X^0$. We arrive at
\begin{equation*}
    X^0 = -(\mu-1)^{-1}\,\p_aX^a(-1)^{\at(\bXt+1)}\,,
\end{equation*}
and $\X$ is given by~\eqref{divfree} (up to a change of notation\,\footnote{\,It is traditional to use German letters for denoting tensor densities of various weights.}). Note that the components $X^a$ can be   arbitrary.
\end{proof}

Therefore, for weight   $\neq 1$, there is a one-to-one correspondence between the divergence-free vector fields $\X$ on $\hat M$ and arbitrary vector densities $\XX$ of the same weight on $M$.

For  an ordinary vector field $X\in \Vect(M)$ (where  $\mu=0$), the corresponding divergence-free vector field on $\hat M$ coincides with the Lie derivative $\X=L_X$.

This motivates the following definition: for a  vector density $\mathfrak{X}\in \Vect(M,\Den)$ of  arbitrary weight $\mu\neq 1$, we call the corresponding divergence-free vector field $\X$ on $\hat M$ defined by equation~\eqref{divfree}, the \emph{generalized Lie derivative} w.r.t.  $\XX$ and denote it $L_{\mathfrak{X}}$.

Now consider derivations of $\Den(M)$ that are not necessarily divergence-free.

Let us return to Example~\ref{divvert}. Vertical vector fields on $\hat M$ are a well-defined subspace of $\Vect(\hat M)$ and the operator
\begin{equation*}
    \X\mapsto \bigl((\w-1)^{-1}\div \X\bigr) \w
\end{equation*}
is a projector on this subspace, defined when $\w-1$ is invertible, i.e., except for weight $1$. The kernel of this projector consists precisely of the divergence-free vector fields on $\hat M$. We can summarize this by the following statement.

\begin{theorem} \label{classder}
For arbitrary vector fields on $\hat M$ with weight $\mu\neq 1$, there is a unique decomposition into the sum of a divergence-free vector field and a vertical vector field, $\X=\X'+\X''$\,, where
  \begin{equation*}
    \X'=X^a\p_a-(\w-1)^{-1}\left(\p_aX^a\,(-1)^{\at(\bXt+1)}\right)\w
  \end{equation*}
is divergence-free (or a generalized Lie derivative) and
\begin{equation*}
    X''=\bigl((\w-1)^{-1}\div \X\bigr) \w
\end{equation*}
is vertical. Here $X^a=X^a(x,t)$. The decomposition makes sense when $\w-1$ is invertible, i.e., for $\X\in\bigoplus_{\mu\neq 1}\Vect_{\mu}(\hat M)$. \qed
\end{theorem}

We can explain this also as follows. Denote by $\Ver_{\mu}(\hat M)$ the space of vertical vector fields of weight $\mu$ on $\hat M$ (in local coordinates $t^{\mu}\ps(x)\,\w$). There is a short exact sequence
\begin{equation*}
    \begin{CD}
    0@>>> \Ver_{\mu}(\hat M)@>i>> \Vect_{\mu}(\hat M) @>p>> \Vect(M,\Den_{\mu}) @>>> 0
\end{CD}
\end{equation*}
where $i\co\Ver_{\mu}(\hat M)\to\Vect_{\mu}(\hat M)$ is the natural inclusion and $p\co\Vect_{\mu}(\hat M) \to \Vect(M,\Den_{\mu})$ is the natural projection sending a vector field $\X$ to its restriction onto the subalgebra $\fun(M)$. For $\mu\neq 1$, this sequence splits and we have the direct sum decomposition
\begin{equation*}
    \Vect_{\mu}(\hat M) = \Ver_{\mu}(\hat M)\oplus \Vect(M,\Den_{\mu})\,.
\end{equation*}
The splitting is given by the maps $\Vect(M,\Den_{\mu}) \to \Vect_{\mu}(\hat M)$ sending $\XX$ to $L_{\XX}$ and  $\Vect_{\mu}(\hat M)\to \Ver_{\mu}(\hat M)$ sending $\X$ to $\bigl((\mu-1)^{-1}\div\X\bigr)\w$\,.

\smallskip
{\small
The case of weight $\mu=1$ is exceptional. For $\mu=1$, on vector densities   on $M$ there is a canonical  divergence  $\div\co \Vect(M,\Den_1) \to \fun(M)$, $\div\XX=\p_a \XX^a (-1)^{\at(\XXt+1)}$,  and the canonical divergence on $\Vect_1(\hat M)$ factors through the projection on $\Vect(M,\Den_1)$,   $\div=\div\circ\, p$, so   the splitting  {breaks}
down and we have a ``Jordan block'' situation rather than a direct sum.

}

\smallskip

\begin{comment}
\begin{tikzcd}
A \arrow[hook]{r}{u}[swap]{b}
\arrow[two heads]{rd}{u}[swap]{b}
&B \arrow[dotted]{d}{r}[swap]{l}
\arrow[hookleftarrow]{r}{u}[swap]{b}
&C \arrow[two heads]{ld}{b}[swap]{u}\\
&D
\end{tikzcd}
\end{comment}

(Theorem~\ref{classdivfree} and Theorem~\ref{classder} were obtained in~\cite{tv:laplace2}.)

\subsection{Corollary: the Lie bracket of vector densities}
The commutator of two di\-ver\-gence-free derivations of $\Den(M)$ (or vector fields on $\hat M$) is again divergence-free. Indeed, such vector fields are exactly those whose flows preserve the volume element, and the same holds for the flow generated by  the commutator.

Alternatively, we may apply the following statement.

\begin{proposition} The canonical divergence $\div$ on $\Vect(\hat M)$ satisfies the identity
\begin{equation}\label{divcommut}
    \div [\X,\Y]= \X (\div \Y)-(-1)^{\bXt\bYt}\Y(\div\X)\,.
\end{equation}
\end{proposition}
\begin{remark}
Such an identity holds for any abstract divergence defined by an invariant scalar product and expresses its `flatness'. In   general,  there is   an extra   term (`curvature').
\end{remark}
\begin{proof} We use the definition of divergence on $\Vect(\hat M)$ in the form $\div\X=-\X-\X^*$. Following~\cite{tv:laplace2}, we may extend this formula to arbitrary differential operators by setting
\begin{equation*}
    \div \D:=-(\D-(-1)^k\D^*)
\end{equation*}
for an operator  of order $\leq k$. The operation $\div$ so defined takes operators of order $\leq k$ to   operators of order $\leq k-1$. The following identity is obtained by a direct check:
\begin{equation*}
    \div [\D_1,\D_2]=[\div \D_1, \D_2]+[\D_1,\div \D_2] +[\div \D_1,\div \D_2]\,.
\end{equation*}
Since functions commute, equation~\eqref{divcommut} for   vector fields follows as a special case.
\end{proof}
\begin{remark} As noted in~\cite{tv:laplace2}, the operation $\div$ on differential operators satisfies $\div^2=0$ and thus defines a complex.
\end{remark}

Denote the Lie superalgebra of divergence-free vector fields on $\hat M$ by $\SVect(\hat M)=\bigoplus_{\mu} \SVect_{\mu}(\hat M)$. By Theorem~\ref{classdivfree}, $\SVect_{\mu}(\hat M)\cong \Vect(M,\Den_{\mu})$\,,
for $\mu\neq 1$.

Hence, the commutator of vector fields on $\hat M$ induces the \emph{Lie bracket of vector densities} on $M$ as follows. Given $\XX \in \Vect(M,\Den_{\mu})$ and $\YY \in \Vect(M,\Den_{\nu})$, where $\mu, \nu\neq 1$, take the commutator of the corresponding generalized Lie derivatives $L_{\XX}$ and $L_{\YY}$, and restrict it back to $M$\,:
\begin{equation*}
    [\XX,\YY]:=\left.[L_{\XX},L_{\YY}]\right|_{\fun(M)}\,.
\end{equation*}

This gives the following explicit formula. If $\XX=|Dx|^{\mu}\,\XX^a(x)\p_a$ and $\YY=|Dx|^{\nu}\,\YY^a(x)\p_a$\,, then

\begin{multline}\label{liebracketdens}
    [\XX,\YY]=|Dx|^{\mu+\nu}\, \left(\Bigl(\XX^a\p_a -\frac{\nu}{\mu-1}\,\p_a\XX^a(-1)^{\at(\XX+1)}\Bigr)\YY^b  \right.\\
    -\left.(-1)^{\XXt\YYt}\Bigl(\YY^a\p_a -\frac{\mu}{\nu-1}\,\p_a\YY^a(-1)^{\at(\YY+1)}\Bigr)\XX^b\right)\p_b\,.
\end{multline}

%By the definition, \begin{equation*}    L_{[\XX,\YY]}=[L_{\XX},L_{\YY}]\,. \end{equation*}

\subsection{(Anti-)self-adjoint operators of the first and second order}

For differential operators $L$ of order $\leq k$ on the algebra $\Den(M)$, it makes sense to study the condition   $L^*=\pm\, L$  depending on the parity of the number $k$. (That is,  $\div L=0$, for a formal `divergence on operators' defined above.)

A first-order differential operator $L$ on the algebra $\Den(M)$ is anti-self-adjoint ($L^*=-L$) if and only if it is the generalized Lie derivative along a vector density: $L=L_{\mathfrak{X}}$, for $\mathfrak{X}\in\Vect(M,\Den_{\mu})$. To it there corresponds a pencil
\begin{equation*}
    L_{\la}= |Dx|^{\mu}\left(\XX^a\p_a-\frac{\la}{\mu-1}\, \p_a\XX^a\,(-1)^{\at(\XXt+1)}\right)\,.
\end{equation*}
In view of the definition of $\div$, this is a re-statement of Theorem~\ref{classdivfree}.

\begin{remark}
As we showed in~\cite{tv:laplace2}, a second-order  differential operator   $L$ on the algebra $\Den(M)$ which is self-adjoint ($L^*=+L$) and normalized by the condition $L1=0$ corresponds to an operator pencil of the form
\begin{multline*}
    L_{\la}=|Dx|^{\mu}\frac{1}{2}\Biggl(\SSS^{ab}\p_b\p_a+
    \left(\p_b\SSS^{ba}(-1)^{\bt(\e+1)}+(2 \la+\mu-1)\g^a\right)\p_a +\\
     \la\,\p_a\g^a(-1)^{\at(\e+1)}  +
         \la(\la+ \mu-1)\,\lt \Biggr).
\end{multline*}
It is completely defined by the following geometric data: the principal symbol $\SSS^{ab}$, which is a tensor density on $M$; the subprincipal symbol  $\g^a$, which has the geometric meaning of an `upper connection' on volume forms associated with $\SSS^{ab}$;  and a so-called `Brans--Dicke field' $\lt$\,. Here $\e=\tilde L$\,.
\end{remark}

\section{Generalization to multivector fields}

\subsection{Recollection of multivector fields} For a (super)manifold $M$, by $\A(M)$ we denote the algebra of multivector fields\,\footnote{In the super case, we consider, strictly speaking,   `pseudo-multivector fields'.} on $M$. Recall some basic facts concerning this algebra.
\begin{itemize}
  \item $\A(M)=\fun(\Pi T^*M)$, where $\Pi$ is the parity reversion functor. (We can take this as the definition of $\A(M)$.) In local coordinates, a multivector field has the form $P=P(x,x^*)$, where    $x^*_a$ transforms as $\p_a$ and $\tilde x^*_a=\at+1$ (the `antimomentum' conjugate to $x^a$);
  \item On the algebra $\A(M)$ there is a canonical Schouten bracket (odd Poisson bracket);
  \item On  the module of multivector densities $\A(M, \Vol)$ there is a canonical odd Laplacian
  \begin{equation*}
    \delta =\der{}{x^a}\der{}{x^*_a}
  \end{equation*}
  (or `divergence of multivector densities')
 and $\delta^2=0$;
  \item For any choice of a volume element, allowing to identify $\A(M, \Vol)$ with $\A(M)$\,, the corresponding odd Laplacian on $\A(M)$ generates the Schouten bracket:
      \begin{equation*}
        \delta(PQ)=\delta P\, Q+(-1)^{\Pt}P\,\delta Q +(-1)^{\Pt+1}[P,Q]
      \end{equation*}
and it is a derivation of the bracket:
\begin{equation*}
    \delta [P,Q]=[\delta P,Q]+(-1)^{\Pt+1}[P,\delta Q]\,.
\end{equation*}
\end{itemize}
Details concerning these facts can be found, for example, in~\cite{tv:graded} and \cite{tv:laplace1}.

\begin{remark} Different names for $\d$ (`divergence' and `Laplacian') should not cause confusion. The operator $\d$ on $\A(M,\Vol)$ is called a `divergence' because it extends the divergence on vector densities $\Vect(M,\Vol)$. On the other hand, it is an (odd) `Laplacian' because it is a second-order differential operator from the viewpoint of $\Pi T^*M$.
\end{remark}

{\small

\begin{remark} For supermanifolds, the complex $\bigl(\A(M, \Vol),\d\bigr)$ is known as the complex of \emph{(pseudo)\-integral forms} and the notation $\Sigma(M)$ for $\A(M,\Vol)$ is employed. In the ordinary case, multivector densities are isomorphic to differential forms and $\d$ is just a `dual' description of the de Rham differential, but in the super case, there is no
such   isomorphism (differential and integral forms are different objects; they are particular instances of the `super de Rham complexes' $\O^{*|s}(M)$, $s=0,\ldots,m$, if $M=M^{n|m}$). For various aspects of that see, e.g., \cite{tv:git, tv:dual, tv:susy, tv:ber}.
\end{remark}

}

\subsection{Classification of multivector fields on $\hat M$}

We shall consider graded multivector fields on the (super)manifold $\hat M$ for a given (super)manifold $M$ and obtain their classification in the way similar to the above classification of vector fields. With an abuse of language, we shall speak about `multivector fields' on $\hat M$ and for the algebra $\Den(M)$ interchangeably. On $\hat M$, we have $\A(\hat M)=\oplus_{\mu}\A_{\mu}(\hat M)$.

Let  $P\in \A(\hat M)$. In local coordinates, we may write $P=P(x,t,x^*,t^*)$ as
\begin{equation}\label{multivect}
    P=P_0(x,t,x^*)+tt^*P_1(x,t,x^*)\,.
\end{equation}
Here the variables $x^*_a$ are the conjugate antimomenta for $x^a$ and the  variable $t^*$ is the conjugate antimomentum for $t$. Note that $t^*$ is odd and $w(t^*)=-1$. Note also the transformation laws
\begin{align}
    x^*_a&= \p_a{x^{a'}}  x^*_{a'} +  (\p_a\ln J)\, t't'^*\,, \label{xstaronhatm} \\
    tt^* &= t't'^*  \,.
\end{align}
(compare~\eqref{translawforda}, \eqref{translawfordt}). Thus, setting $t^*=0$ to isolate $P_0$ in~\eqref{multivect} makes invariant sense.

On $\hat M$, we have the canonical volume element $|D(x,t)|t^{-2}$.  Therefore on multivector fields $\A(\hat M)$ there is the corresponding \emph{canonical odd Laplace operator} (or \emph{canonical divergence}):
\begin{equation}
    \delta= \der{}{x^a}\der{}{x^*_a} + t^{2}\der{}{t}t^{-2}\der{}{t^*}\,.
\end{equation}
Hence, for $P\in\A(\hat M)$ given in coordinates by~\eqref{multivect}, we obtain
\begin{multline*}
    \delta P = \left(\der{}{x^a}\der{}{x^*_a} + t^{2}\der{}{t}t^{-2}\der{}{t^*}\right)\bigl(P_0(x,t,x^*)+tt^*P_1(x,t,x^*)\bigr)
     = \\
    \delta_0 P_0 -tt^*\delta_0P_1 +t^2\der{}{t}t^{-1}P_1=
    \delta_0 P_0 -P_1 +t\der{}{t}P_1 -tt^*\delta_0P_1=
    \delta_0 P_0 -(1-\w)(P_1)   -tt^*\delta_0P_1
    \,,
\end{multline*}
where we denoted $\delta_0=\der{}{x^a}\der{}{x^*_a}$.

Suppose $\delta P=0$. Then it is equivalent to the system:
\begin{align*}
    \delta_0P_0-(1-\w)P_1=0\,,\\
    \delta_0P_1=0\,.
\end{align*}
Since $\d_0^2=0$, the first equation implies $(1-\w)\d_0P_1=0$. The equations can be immediately solved if $1-\w$ is invertible on $P$ (that is, $P$ does not include a term of weight $1$), giving
\begin{equation*}
    P_1= (1-\w)^{-1}\delta_0P_0\,,
\end{equation*}
and we arrive at the following classification theorem (analogous to Theorem~\ref{classder} for derivations).
\begin{theorem}[A. Biggs, 2013] \label{classmult}
Every divergence-free multivector field on $\hat M$ on which $1-\w$ is invertible has the form
\begin{equation*}
    P=P_0 +tt^* (1-\w)^{-1}\delta_0P_0\,,
\end{equation*}
where $\delta_0=\der{}{x^a}\der{}{x^*_a}$.
It is completely defined by the first term $P_0=P_0(x,t,x^*)$, which is a  multivector density on $M$. \qed
\end{theorem}
Therefore, for $\mu\neq 1$,
\begin{equation*}
    \Ker \{\d\co \A_{\mu}(\hat M)\to \A_{\mu}(\hat M)\} \cong \A(M,\Den_{\mu})\,.
\end{equation*}
(The case $\mu=1$ is again exceptional. In this case, the operator $\d_0$ is invariant on $\A(M,\Vol)$ and $\d P=0$ is equivalent to the two conditions: $\d_0P_0=0$ and $\d_0P_1=0$.)

\subsection{Useful corollaries. Application to Poisson brackets}
Since the divergence (corresponding to a choice of a volume element) is a derivation of the Schouten bracket for any manifold, the divergence-free multivector fields are closed under the bracket\,\footnote{Worth mentioning that the bracket of such fields is always $\d$-exact, $[P,Q]=\pm\,\d(PQ)$, so it induces zero on $\d$-cohomology.}. In particular, this holds for the canonical divergence $\d$ on $\hat M$. We have a Lie superalgebra
\begin{equation*}
    \SA(\hat M):=\Ker \{\d\co \A(\hat M)\to \A(\hat M)\}= \bigoplus_{\mu} \SA_{\mu}(\hat M)\,.
\end{equation*}
(This is not a Poisson subalgebra of $\A(\hat M)$, because the product of divergence-free multivector fields need not be divergence-free; in fact, its divergence is up to sign the bracket.)

By Theorem~\ref{classmult}, an odd Lie bracket is induced on multivector densities on $M$. Suppose $\PP\in\A(M,\Den_{\mu})$, where $\mu\neq 1$. To it corresponds the divergence-free multivector field $\hat\PP$ on $\hat M$,
\begin{equation*}
    \hat\PP=\PP  +tt^*\, \frac{1}{1-\mu}\, \delta_0\PP
\end{equation*}
(this is analogous to the definition of the `generalized Lie derivative' for vector densities). If $\QQ\in\A(M,\Den_{\nu})$, where $\nu\neq 1$, we have the induced odd bracket
\begin{equation*}
    [\PP,\QQ]:= [\hat\PP,\hat\QQ]_{|t^*=0}\,.
\end{equation*}
Explicitly,
\begin{equation}\label{bracketmult}
    [\PP,\QQ]=[\PP,\QQ]_0+\frac{\nu}{1-\mu}\,(-1)^{\tilde\PP+1}\d_0\PP\cdot\QQ - \frac{\mu}{1-\nu}\,\PP\cdot\d_0\QQ
\end{equation}
(by a direct check). At the r.h.s. we have denoted by $[-,-]_0$ the Schouten bracket on $M$ applied    to multivector densities, which can be done in particular coordinates; neither it  nor  the operator $\d_0$ have  an intrinsic  meaning, but the whole bracket $[-,-]$ does.

We see, in particular, that if $\mu=\nu=0$ in equation~\eqref{bracketmult}, the last two terms disappear and the odd bracket induced from $\hat M$ coincides with the canonical Schouten bracket on $M$. Hence multivector fields on $M$ can be lifted to multivector fields on $\hat M$ by
\begin{equation}\label{hatp}
    \hat P=P  +tt^*\,  \delta_0 P\,,
\end{equation}
where $P\in\A(M)$,   and this lifting preserves the Schouten brackets.

\begin{remark} If $P=P(x,x^*)$ is a multivector field on $M$, one cannot extend it to $\hat M$ ``identically'' by regarding the antimomentum variables $x^*_a$, which are conjugate to $x^a$ on $M$, as the same conjugate variables on $\hat M$. They transform differently (see~\eqref{xstaronhatm} for the transformation law on $\hat M$) and such an ``identical'' lifting would have no invariant meaning. The second term in~\eqref{hatp} compensates for that.
\end{remark}

\begin{corollary} \label{corol}
Every even Poisson structure on $M$ extends canonically to an even Poisson bracket on the algebra of densities $\Den(M)$. The same holds for homotopy Poisson structures.
\end{corollary}

\begin{example} Suppose $P\in\A(M)$,
\begin{equation*}
    P=\frac{1}{2}\,P^{ab}(x)x^*_bx^*_a\,.
\end{equation*}
Then $\hat P\in\A(\hat M)$ is given by
\begin{equation*}
    \hat P=\frac{1}{2}\,P^{ab}(x)x^*_bx^*_a+tt^* \p_aP^{ab}\,x^*_b=\frac{1}{2}\,P^{ab}(x)x^*_bx^*_a-t \p_aP^{ab}\,x^*_bt^*\,.
\end{equation*}
If $P$ is a Poisson tensor, i.e., $[P,P]=0$, then $\hat P$ is also a Poisson tensor, and for the `lifted' Poisson bracket on $\Den(M)$ we obtain, by a direct check,
\begin{equation*}
    \{x^a,x^b\}_{\hat P}={\{x^a,x^b\}_{P}=}(-1)^{\at}P^{ab}(x)\,, \quad
    \{t,x^b\}_{\hat P}=-t\p_aP^{ab}(x)\,, \quad
    \{t,t\}_{\hat P}=0\,.
\end{equation*}
(Recall   $\{\ps,\ch\}_{\hat P}=\bigl[[\ps,\hat P],\ch\bigr]$. See~\cite{tv:graded} for general reference on even and odd brackets.)
\end{example}

(The statement of Corollary~\ref{corol} was  obtained  by A.~Biggs, which    led him to Theorem~\ref{classmult}, see~\cite{biggs:lifting}.)

\section{Final remarks. Analogy with the Nijenhuis bracket}%% maybe move to introduction

\subsection{`Invariant operators on geometric  quantities'}

Such operators have been studied and   classification theorems were obtained (see Kirillov~\cite{kirillov:invariant}). In particular, there is a complete classification of  invariant binary operations on tensor densities, at least in the case of ordinary manifolds (a theorem of P.~Grozman, see~\cite[Theorem 5]{kirillov:invariant}). Therefore one cannot expect to discover new unknown operations here. The principal interest of the Lie brackets on vector   or multivector densities introduced in the previous sections is that they arise as consequences of  Theorems~\ref{classdivfree} and \ref{classmult}.
These theorems describe in terms of the original (super)manifold the derivations or multivector fields for an algebra naturally associated with it, the algebra of densities $\Den(M)$. Such derivations or multivector fields possess   extra structures, the divergence operators $\div$ or $\d$, natural with respect to the manifold $M$.

The situation here is very similar to the classical results due to Nijenhuis and we shall briefly discuss this analogy.

\subsection{Nijenhuis's classification of derivations of $\O(M)$}
Consider the algebra of forms\,\footnote{In the super case, more precisely, we speak about pseudodifferential forms.} $\O(M)$ on a (super)manifold $M$. We may identify $\O(M)$ with the algebra of functions  $\fun(\Pi TM)$ on the antitangent bundle $\Pi TM$. Derivations of $\O(M)$ can be identified with the vector fields on the supermanifold $\Pi TM$. As local coordinates on $\Pi TM$ we may take $x^a,dx^a$ (here $dx^a$ are regarded as commuting variables of parity opposite to that of $x^a$). The general form of a vector field $X\in\Vect(\Pi TM)$ is
\begin{equation}\label{vectonpitm}
    \X=X^a_0(x,dx)\,\der{}{x^a}+X^a_1(x,dx)\,\der{}{dx^a}\,.
\end{equation}
In particular, one can write in this form the exterior differential,
\begin{equation*}
    d=dx^a\,\der{}{x^a}\,,
\end{equation*}
the interior product $i_X$ with a vector field $X$ on $M$, $X=X^a(x)\p_a$,
\begin{equation*}
    i_X=(-1)^{\Xt}X^a(x)\,\der{}{dx^a}\,,
\end{equation*}
and the Lie derivative $L_X$ w.r.t. such a vector field $X\in\Vect(M)$,
\begin{equation*}
    L_X=X^a(x)\,\der{}{x^a}+(-1)^{\Xt}dX^a(x) \,\der{}{dx^a}\,.
\end{equation*}
(The last expression can be obtained directly from considering the infinitesimal shift on $\Pi TM$ generated by $X$ or from the formula $L_X=[d,i_X]$.)

Note that, in this language, $d$ is a distinguished odd vector field on $\Pi TM$, with the property $[d,d]=2d^2=0$.

There is  a classical statement:
\begin{theorem}[Nijenhuis and Fr\"{o}licher--Nijenhuis]
Every derivation of $\O(M)$ (= vector field on $\Pi TM$) that commutes with $d$ is completely defined by its restriction on the subalgebra $\fun(M)\subset \O(M)$, which is an element of $\Vect(M,\O)$. Explicitly:
\begin{equation*}
     \X=X^a(x,dx)\der{}{x^a}+(-1)^{\Xt} dX^a(x,dx)\der{}{dx^a}\,,
\end{equation*}
where $X=X^a(x,dx)\p_a\in \Vect(M,\O)$\,.
\end{theorem}

{\small
\begin{proof} Consider $[d,\X]$ where $\X$ is given by~\eqref{vectonpitm}. We obtain
\begin{equation*}
    [d,\X]=dX^a_0\,\der{}{x^a}+dX^a_1\,\der{}{dx^a}-(-1)^{\bXt}X^a_1\,\der{}{x^a}=
    \left(dX^a_0-(-1)^{\bXt}X^a_1\right)\der{}{x^a}+dX^a_1\,\der{}{dx^a}\,,
\end{equation*}
hence $[d,\X]=0$ is equivalent to the system $dX^a_0-(-1)^{\bXt}X^a_1=0$ and $dX^a_1=0$, which reduces to the single equation $X^a_1= (-1)^{\bXt}dX^a_0$ with no restriction on $X^a(x,dx)$.
\end{proof}

}

The derivations commuting with $d$ are interpreted as   `generalized Lie derivatives'. They form a  subalgebra in the Lie superalgebra of all vector fields on $\Pi TM$. The commutator of vector fields belonging to this subalgebra gives an isomorphic operation on the space of vector-valued forms $\Vect(M,\O)$. It is called the \emph{Nijenhuis bracket}. Explicitly, for $X=X^a(x,dx)\p_a)$ and $X=Y^a(x,dx)\p_a$, their Nijenhuis bracket is
\begin{equation*}
    [X,Y]=\left.[L_X,L_Y]\right|_{\fun(M)}
\end{equation*}
and
\begin{equation*}
    [X,Y]=\left(\bigl(X^a\der{}{x^a}+(-1)^{\Xt}dX^a\der{}{dx^a}\bigr)Y^b-(-1)^{\Xt\Yt}
    \bigl(Y^a\der{}{x^a}+(-1)^{\Yt}dY^a\der{}{dx^a}\bigr)X^b\right)\p_b
\end{equation*}

We can see a clear analogy with Theorems~\ref{classdivfree} and \ref{classmult}, and the constructions of brackets on $\Vect(M,\Den)$ and $\A(M,\Den)$\,.

\begin{remark}[History of bracket operations] Schouten defined a binary operation on multivector fields in~\cite{schouten:1940}. It has not been appreciated as a Lie-type bracket until Nijenhuis~\cite{nijen:jacobi} proved an analog of the Jacobi identity for it. This was a fundamental  example of what later became Lie superalgebras (another important example that arose about the same time was the Whitehead bracket on homotopy groups). In the same paper, Nijenhuis introduced his bracket for vector-valued forms (of form-valued vector fields). It was elaborated in Fr{\"o}licher--Nijenhuis~\cite{nijen:der1}. Of course, they were not using the supermanifold language. It should be noted that Nijenhuis was very close to the ideas of supermanifolds and supergroups. In particular, Nijenhuis and Richardson (see~\cite{nijen:deform} and subsequent papers), for the needs of algebraic deformation theory, introduced  pairs $(\mathfrak g, G_0)$ consisting of a Lie superalgebra $\mathfrak g=\mathfrak g_0\oplus \mathfrak g_1$ and a Lie group $G_0$ with the Lie algebra $\mathfrak g_0$, which served as an effective replacement of  (not yet known  {at that time})  Lie supergroups\,\footnote{Such pairs are sometimes referred to as `Harish-Chandra pairs' nowadays, following an analogy with  {different}  objects considered by Harish-Chandra in representation theory, but `Nijenhuis' or  `Nijenhuis--Richardson' pairs  would be  {a} much more historically justified name. (See more  {details on the
history} in~\cite{karabegov:knv}).}.

\end{remark}

%\bibliographystyle{plain}
%\bibliography{geometry}
%\end{document}

\end{document}